\documentclass[10pt, conference]{IEEEtran}

\usepackage{amsmath,amssymb,amsthm}
\usepackage{graphicx}
\usepackage{subfigure}
\usepackage{curves}
\usepackage{amsfonts}
\usepackage{epsfig}
\usepackage{stfloats}
\usepackage{cite}
\usepackage[ruled,vlined]{algorithm2e}
\usepackage{epstopdf}
\usepackage{multirow}
\usepackage{array}
\usepackage{rotating}
\usepackage{comment}

\begin{document}
	
	\title{On the Polarization of R\'{e}nyi Entropy}
	\author{\IEEEauthorblockN{Mengfan~Zheng}
			\IEEEauthorblockA{Imperial College London, \\United Kingdom\\Email: m.zheng@imperial.ac.uk}
		\and
		\IEEEauthorblockN{Ling~Liu}
		\IEEEauthorblockA{Huawei Technologies Co. Ltd.\\Shenzhen, P. R. China\\Email: liuling\_88@pku.edu.cn}
		\and
		\IEEEauthorblockN{Cong~Ling}
		\IEEEauthorblockA{Imperial College London, \\United Kingdom\\Email: c.ling@imperial.ac.uk}
	}

	\maketitle
	
	\begin{abstract}
		Existing polarization theories have mostly been concerned with Shannon's information measures, such as Shannon entropy and mutual information, and some related measures such as the Bhattacharyya parameter. In this work, we extend polarization theories to a more general information measure, namely, the R\'{e}nyi entropy. Our study shows that under conditional R\'{e}nyi entropies of different orders, the same synthetic sub-channel may exhibit opposite extremal states. This result reveals more insights into the polarization phenomenon on the micro scale (probability pairs) rather than on the average scale.
		
	\end{abstract}
	
	
	\section{Introduction}
	\label{Sec:Intro}
	The polarization technique (including channel polarization \cite{arikan2009channel} and source polarization \cite{arikan2010source}) is one of the most significant breakthrough in information theory over the past decade. Ar{\i}kan showed us that as the size of the polar transformation goes to infinity, the conditional entropies of the synthetic sub-channels (or random variable pairs) equal 0 or 1 almost everywhere (a.e.) \cite{arikan2009channel}. Also, their varentropies (variance of the conditional entropy random variable) asymptotically decrease to zero \cite{arikan2016varentropy}. These results imply that the sub-channels' transition probability matrices tend to either deterministic (noiseless channels) or uniform with respect to any channel input (completely noisy channels). However, are they still close to uniform or deterministic distributions under stricter criteria? Polarization results using Shannon's information measures fail to answer this question.	
	
	In this work, we study polarization using a more general information measure, i.e., the R\'{e}nyi entropy \cite{renyi1961measures}. The R\'{e}nyi entropy is more sensitive to deviations from uniform or deterministic distributions, and has been used in many areas where Shannon entropy may not be a good metric. For example, the collision entropy and min-entropy, both of which are special cases of the R\'{e}nyi entropy, are convenient metrics for privacy amplification in secret-key agreement \cite{bloch2011physical}. 
	The R\'{e}nyi entropy of a random variable $X$ is defined as follows.	
	\newtheorem{definition}{Definition}
	\begin{definition}[R\'{e}nyi Entropy \cite{renyi1961measures}]
		\label{def:RenyiEntropy}
		The R\'{e}nyi entropy of a random variable $X\in\mathcal{X}$ of order $\alpha$ is defined as
		\begin{equation}
		H_{\alpha}(X)=\frac{1}{1-\alpha}\log \sum_{x\in\mathcal{X}}{P_X(x)^{\alpha}},
		\end{equation}
	\end{definition}
    It can be shown that as $\alpha\rightarrow 1$, the R\'{e}nyi entropy reduces to the Shannon entropy. Two other special cases of the R\'{e}nyi entropy which will be discussed later include the \textit{max-entropy}:
    \begin{equation}
    	H_{0}(X)=\log |\mathcal{X}|, \label{max-entropy}
    \end{equation}
    which is the R\'{e}nyi entropy of order 0, and the \textit{min-entropy}:
    \begin{equation}
    H_{\infty}(X)=\min_{i}(-\log p_i)=-\log\max_i p_i,
    \end{equation}
    which equals the limiting value of $H_{\alpha}(X)$ as $\alpha\rightarrow \infty$.
	
	Unlike the conditional Shannon entropy, there is no generally accepted definition of the conditional R\'{e}nyi entropy yet. In this paper, we adopt the following definition of conditional R\'{e}nyi entropy in the study of polarization.
	
	\begin{definition}[conditional R\'{e}nyi entropy \cite{jizba2004world,golshani2009some}]
		\label{def:CondRenyiEntropyJiz}
		The conditional R\'{e}nyi entropy of order $\alpha$ of $X$ given $Y$ is defined as
		\begin{align}
		H_{\alpha}(X|Y)&=\frac{1}{1-\alpha}\log \frac{\sum_{\{x,y\}\in\mathcal{X}\times\mathcal{Y}}{P_{X,Y}(x,y)^{\alpha}} }{\sum_{y\in\mathcal{Y}}P_Y(y)^{\alpha}}. \label{def:CondiRenyi}
		\end{align}
	\end{definition}
    Note that this type of conditional R\'{e}nyi entropy satisfies the chain rule:
    \begin{equation}
    H_{\alpha}(X|Y)+H_{\alpha}(Y)=H_{\alpha}(X,Y).\label{RenyiCondChain}
    \end{equation}
	This means that it reduces to the conditional Shannon entropy when $\alpha=1$. 

	
	

	There have been very limited researches in regard to polarization of conditional R\'{e}nyi entropies. In \cite{Alsan2014E0} it is shown that the following chain rule inequality holds for the polar transformation for $\alpha \leq 1$,
	\begin{align*}
	H^*_{\alpha}(U_1U_2|Y_1Y_2)\geq H^*_{\alpha}(U_1|Y_1Y_2)+H^*_{\alpha}(U_2|Y_1Y_2U_1),
	\end{align*}
	whenever $U_1$, $U_2$ are i.i.d. uniform on $\mathbb{F}_2$. The inequality holds with equality if and only if the channel $W$ is perfect, or the channel $W$ is completely noisy, or $\alpha = 1$. Note that $H^*_{\alpha}(X|Y)$ in \cite{Alsan2014E0} is defined as $H^*_{\alpha}(X|Y)=\frac{\alpha}{1-\alpha}\log \sum_{y\in\mathcal{Y}}P_Y(y) \Big{[}\sum_{x\in\mathcal{X}}{P_{X|Y}(x|y)^{\alpha}\Big{]}^{\frac{1}{\alpha}}}$.
	
	In this paper, we also restrict ourselves to the binary case (i.e., $\mathcal{X}=\mathbb{F}_2$) as a starting point. As a result, logarithms in this paper will all be base-2. However, we do not assume $X$ to be uniformly distributed. 
	We prove that the order-$\alpha$ conditional R\'{e}nyi entropies of the synthetic sub-channels also polarize to 0 and 1, but the fraction of 1 tends to $H_{\alpha}(X|Y)$. This is to say that for a given sub-channel, its conditional R\'{e}nyi entropies of different orders may exhibit opposite extremal states. Intuitively, if a sub-channel is truly noiseless (or truly completely noisy), its conditional R\'{e}nyi entropies of different orders should all be 0 (or 1). We show both analytically and numerically that this strange phenomenon is caused by a vanishing deviation from truly deterministic or truly uniform distributions. The different extremal states that a sub-channel exhibits for various $\alpha$ reflect the polarization level of the sub-channel at the micro scale, i.e., how close is its joint distribution to truly uniform or truly deterministic.

	\section{Polarization of R\'{e}nyi Entropy}
	
	\subsection{Polarization of a Basic Polar Transformation}
	
	Consider two binary-input discrete memoryless channels (B-DMC) $P_{Y_1|X_1}$ and $P_{Y_2|X_2}$ with the same input distribution $P_X$ in the channel coding scenario, or consider $(X_1,Y_1)$ and $(X_2,Y_2)$ as two samples of a memoryless source $(X,Y)\sim P_{X,Y}$ with $X$ being the binary source to be compressed and $Y$ being the side information about $X$ in the source coding scenario. Note that in the former case, $P_{Y_1|X_1}$ and $P_{Y_2|X_2}$ can be different, which corresponds to the compound channel setting. Let
	\begin{align}
	U_1=X_1\oplus X_2,~~U_2=X_2,\label{PolarTransform}
	\end{align}
	and denote 
	\begin{align*}
	&~~~~P_A(u_1,u_2,y_1,y_2)\\
	&\triangleq P_{U_1,U_2}(u_1,u_2)P_{Y_1|X_1}(y_1|u_1\oplus u_2)P_{Y_2|X_2}(y_2|u_2)
	\end{align*}
	for short. From (\ref{PolarTransform}) we know that
	\begin{equation}
	\label{PU1U2}
	\begin{aligned}
	P_{U_1,U_2}(0,0)=P_{X}(0)P_{X}(0)&,~
	P_{U_1,U_2}(0,1)=P_{X}(1)P_{X}(1),\\
	P_{U_1,U_2}(1,0)=P_{X}(1)P_{X}(0)&,~
	P_{U_1,U_2}(1,1)=P_{X}(0)P_{X}(1).
	\end{aligned}
	\end{equation}
	For the basic polar transformation, we have the following polarization result.
	
	\newtheorem{lemma}{Lemma}
	\begin{lemma}
		\label{lemma:RenyiPolar1}
		For $\alpha\geq 0$, we have
		\begin{align}
		H_{\alpha}(U_2|Y_1Y_2U_1) &\leq \min \{H_{\alpha}(X_1|Y_1), H_{\alpha}(X_2|Y_2)\} \label{RenyiPolar-a1}.\\
		H_{\alpha}(U_1|Y_1Y_2) &\geq \max \{H_{\alpha}(X_1|Y_1), H_{\alpha}(X_2|Y_2)\} ,\label{RenyiPolar-a2}\\
		H_{\alpha}(U_1U_2|Y_1Y_2)&=H_{\alpha}(U_1|Y_1Y_2)+H_{\alpha}(U_2|Y_1Y_2U_1).  \label{RenyiPolar-a3}
		\end{align}	
	\end{lemma}
	
	Let us first recall the Minkowski inequality before proving Lemma \ref{lemma:RenyiPolar1}. 
	\newtheorem{proposition}{Proposition}
	\begin{proposition}
		\label{Prop:Minkowski}
		The Minkowski inequality states that for $1\leq p\leq \infty$,
		\begin{align}
		\Bigg{(} \sum_{k=1}^{n}|x_k+y_k|^p\Bigg{)}^{\frac{1}{p}}\leq \Bigg{(} \sum_{k=1}^{n}|x_k|^p\Bigg{)}^{\frac{1}{p}}+\Bigg{(} \sum_{k=1}^{n}|y_k|^p\Bigg{)}^{\frac{1}{p}}, \label{MinkovIneq}
		\end{align}
		with equality if and only if $\mathbf{x}=(x_1,x_2,...,x_n)$ and $\mathbf{y}=(y_1,y_2,...,y_n)$ are positively linearly dependent, i.e., $\mathbf{x}=\lambda \mathbf{y}$ for some $\lambda\geq 0$ or $\mathbf{y}=0$. For $0<p< 1$, the inequality in (\ref{MinkovIneq}) is reversed.
	\end{proposition}

	\begin{proof}[Proof of Lemma \ref{lemma:RenyiPolar1}]

		(I) First, we consider (\ref{RenyiPolar-a3}). Denote
		\begin{align}
		S_1&=\sum_{\{y_1,y_2\}\in\mathcal{Y}^2}\Big{(}\sum_{\{u_1,u_2\}\in\mathcal{X}^2}P_A(u_1,u_2,y_1,y_2)\Big{)}^{\alpha}\label{S1}\\
		S_2&=\sum_{\{y_1,y_2\}\in\mathcal{Y}^2}\sum_{u_1\in\mathcal{X}}\Big{(}\sum_{u_2\in\mathcal{X}}P_A(u_1,u_2,y_1,y_2)\Big{)}^{\alpha}\label{S2}\\
		S_3&=\sum_{\{y_1,y_2\}\in\mathcal{Y}^2}\sum_{\{u_1,u_2\}\in\mathcal{X}^2}\Big{(}P_A(u_1,u_2,y_1,y_2)\Big{)}^{\alpha}\label{S3}
		\end{align}
		Then we have
		\begin{align*}
		&~~~~H_{\alpha}(U_1|Y_1Y_2)+H_{\alpha}(U_2|Y_1Y_2U_1)\\
		&=\frac{1}{1-\alpha}\log \frac{S_2}{S_1}+\frac{1}{1-\alpha}\log \frac{S_3}{S_2}\\
		&=\frac{1}{1-\alpha}\log \frac{S_3}{S_1}\\
		&=H_{\alpha}(U_1U_2|Y_1Y_2).
		\end{align*}

		(II) Next, we consider (\ref{RenyiPolar-a2}).
		\newcounter{mytempeqncnt}
		\begin{figure*}[!t]
			\normalsize
			\begin{align}
			H_{\alpha}(X_1|Y_1)
			&=\frac{1}{1-\alpha}\log \frac{\Big{(}\sum_{y_1\in\mathcal{Y}}\sum_{x_1\in\mathcal{X}}{P_{Y_1,X_1}(y_1,x_1)^{\alpha}}\Big{)}\Big{(}\sum_{y_2\in\mathcal{Y}}\sum_{x_2\in\mathcal{X}}{P_{Y_2,X_2}(y_2,x_2)^{\alpha}}\Big{)} }{\Big{(}\sum_{y_1\in\mathcal{Y}}\big{[}\sum_{x_1\in\mathcal{X}}P_{Y_1,X_1}(y_1,x_1)\big{]}^{\alpha}\Big{)}\Big{(}\sum_{y_2\in\mathcal{Y}}\sum_{x_2\in\mathcal{X}}{P_{Y_2,X_2}(y_2,x_2)^{\alpha}}\Big{)}}=\frac{1}{1-\alpha}\log \frac{S_3}{S_4}\label{Proof:1}
			\end{align}
			\hrulefill
			\vspace*{4pt}
		\end{figure*}		
		$H_{\alpha}(X_1|Y_1)$ can be expressed as (\ref{Proof:1}) (on the top of the next page), where
		\begin{align*}
		S_4&=\sum_{\{y_1,y_2\}\in\mathcal{Y}^2}\Bigg{\{} \Big{[}P_{Y_1,X_1}(y_1,0)+P_{Y_1,X_1}(y_1,1)\Big{]}^{\alpha}\\
		&~~~~~~~~\times\Big{[}P_{Y_2,X_2}(y_2,0)^{\alpha}+P_{Y_2,X_2}(y_2,1)^{\alpha}\Big{]} \Bigg{\}}
		\end{align*}
		From (\ref{S2}) and (\ref{PU1U2}) we have
		\begin{align*}
		S_2&=\sum_{\{y_1,y_2\}\in\mathcal{Y}^2}\Bigg{\{} \Big{[}P_{Y_1,X_1}(y_1,0)P_{Y_2,X_2}(y_2,0)\\
		&~~~~~~~~+P_{Y_1,X_1}(y_1,1)P_{Y_2,X_2}(y_2,1)\Big{]}^{\alpha}\\
		&~~~~~~~~~~~~+\Big{[}P_{Y_1,X_1}(y_1,1)P_{Y_2,X_2}(y_2,0)\\
		&~~~~~~~~~~~~~~~~+P_{Y_1,X_1}(y_1,0)P_{Y_2,X_2}(y_2,1)\Big{]}^{\alpha} \Bigg{\}}.
		\end{align*}
		
		For $0<\alpha<1$, by the Minkowski inequality we have
		\begin{align}
		&\Bigg{\{}\Big{[}P_{Y_1,X_1}(y_1,0)P_{Y_2,X_2}(y_2,0)+P_{Y_1,X_1}(y_1,1)P_{Y_2,X_2}(y_2,1)\Big{]}^{\alpha}\nonumber\\
		&~~~~+\Big{[}P_{Y_1,X_1}(y_1,1)P_{Y_2,X_2}(y_2,0)\nonumber\\
		&~~~~~~~~+P_{Y_1,X_1}(y_1,0)P_{Y_2,X_2}(y_2,1)\Big{]}^{\alpha}\Bigg{\}}^{\frac{1}{\alpha}}\nonumber\\
		&\geq  \Big{\{}\big{[}P_{Y_1,X_1}(y_1,0)P_{Y_2,X_2}(y_2,0)\big{]}^{\alpha}\nonumber\\
		&~~~~+\big{[}P_{Y_1,X_1}(y_1,0)P_{Y_2,X_2}(y_2,1)\big{]}^{\alpha}\Big{\}}^{\frac{1}{\alpha}}\nonumber\\
		&~~~~~~~~+\Big{\{}\big{[}P_{Y_1,X_1}(y_1,1)P_{Y_2,X_2}(y_2,0)\big{]}^{\alpha}\nonumber\\
		&~~~~~~~~~~~~+\big{[}P_{Y_1,X_1}(y_1,1)P_{Y_2,X_2}(y_2,1)\big{]}^{\alpha}\Big{\}}^{\frac{1}{\alpha}} \label{Proof-7}\\
		&=\Big{[}P_{Y_1,X_1}(y_1,0)+P_{Y_1,X_1}(y_1,1)\Big{]}\nonumber\\
		&~~~~\times\Big{[}P_{Y_2,X_2}(y_2,0)^{\alpha}+P_{Y_2,X_2}(y_2,1)^{\alpha}\Big{]}^{\frac{1}{\alpha}}.\nonumber
		\end{align}
		Thus, $S_2\geq S_4$, which means
		\begin{align}
		H_{\alpha}(U_2|Y_1Y_2U_1)&=\frac{1}{1-\alpha}\log \frac{S_3}{S_2}\nonumber\\
		&\leq \frac{1}{1-\alpha}\log \frac{S_3}{S_4}=H_{\alpha}(X_1|Y_1).
		\end{align}
				
		For $\alpha>1$, the inequality in (\ref{Proof-7}) is reversed. However, since $1-\alpha< 0$, the result remains the same.	
		For $\alpha=1$, the Renyi entropy reduces to the Shannon entropy, and the polarization result is identical to the result here \cite{arikan2009channel}.
		For $\alpha=0$, it is obvious from (\ref{max-entropy}) and (\ref{def:CondiRenyi}) that (\ref{RenyiPolar-a3}) holds, while (\ref{RenyiPolar-a1}) and (\ref{RenyiPolar-a2}) hold with equality.
		
		Similarly,
		$H_{\alpha}(U_2|Y_1Y_2U_1)\leq H_{\alpha}(X_2|Y_2)$ can be proved.
		
		(III) Finally, equality (\ref{RenyiPolar-a3}) and inequality (\ref{RenyiPolar-a2}) immediately imply (\ref{RenyiPolar-a1}).

	\end{proof}

	\subsection{Recursive Polar Transformation}
	Now consider extending the basic transformation recursively to higher orders. For $N=2^n$ with $n$ being an arbitrary integer, the recursive transformation can be expressed as 
	$U^{1:N}=X^{1:N}\mathbf{G}_N$ \cite{arikan2009channel},
	where $\mathbf{G}_N=\mathbf{B}_N \textbf{F}^{\otimes n}$ with $\mathbf{B}_N$ being the bit-reversal matrix and $\textbf{F}=
	\begin{bmatrix}
	1 & 0 \\
	1 & 1
	\end{bmatrix}$. Denote $H_N(i)=H_{\alpha}(U^i|Y^{1:N}U^{1:i-1})$. We have the following theorem.
	
	\newtheorem{theorem}{Theorem}
	\begin{theorem}
		\label{Theorem:polarRenyi}
		For any B-DMC $P_{Y|X}$ (or any discrete memoryless source $(X,Y)\sim P_{X,Y}$ over $\mathcal{X}\times \mathcal{Y}$ with $\mathcal{X}=\{0,1\}$ and $\mathcal{Y}$ an arbitrary countable set) and any $\alpha\geq 0$, as $N\rightarrow \infty$ through the power of 2, the fraction of indices $i\in[N]\triangleq \{1,2 ...,N\}$ with $H_N(i)\in (1-\delta,1]$ goes to $H_{\alpha}(X|Y)$, and the fraction with $H_N(i)\in [0,\delta)$ goes to $1-H_{\alpha}(X|Y)$.
		
	\end{theorem}
	
	\begin{proof}
		We follow Ar{\i}kan's martingale approach \cite{arikan2009channel,arikan2014martingale} to complete the proof. First, we introduce the same infinite binary tree as in the proof of \cite[Theorem 1]{arikan2009channel}, with a root node at level 0 and $2^n$ nodes at level $n$. Then define a random walk $\{B_n;n\geq 0\}$ in this tree as follows. The random walk starts at the root node with $B_0=(0,1)$, and moves to one of the two child nodes in the next level with equal probability at each integer time. If $B_n=(n,i)$, $B_{n+1}$ equals $(n+1,2i-1)$ or $(n+1,2i)$ with probability $1/2$ each. Denote $H(0,1)=H_{\alpha}(X|Y)$ and $H(n,i)=H_{\alpha}(U^i|Y^{1:2^n},U^{1:i-1})$ for $n\geq 1$, $i=[2^n]$. Define a random process $\{H_n;n\geq 0\}$ with $H_n=H(B_n)$. It can be shown that the process $\{H_n;n\geq 0\}$ is a martingale due to the chain rule equality of (\ref{RenyiPolar-a3}), i.e.,
		\begin{equation}
		E[H_{n+1}|B_0,B_1,...,B_n]=H_{n}.
		\end{equation}
		
		Since $\{H_n;n\geq 0\}$ is a uniformly integrable martingale, it converges a.e. to an RV $H_{\infty}$ such that $E[|H_{n}-H_{\infty}|]=0$. Then we have 
		\begin{equation}
		E[|H_{n}-H_{n+1}|]\rightarrow 0.\label{Converge-1}
		\end{equation}
		Note that
		\begin{align}
		&~~~~E[|H_{n}-H_{n+1}|]\nonumber\\
		&=\frac{1}{2}\Big{(} E[|H_{\alpha}(U^i|Y^{1:2^n},U^{1:i-1})\nonumber\\
		&~~~~-H_{\alpha}(U^{2i-1}|Y^{1:2^{n+1}},U^{1:2i-2})|]\nonumber\\
		&~~~~~~+E[|H_{\alpha}(U^i|Y^{1:2^n},U^{1:i-1})\nonumber\\
		&~~~~~~~~-H_{\alpha}(U^{2i}|Y^{1:2^{n+1}},U^{1:2i-1})|] \Big{)}.\label{Converge-2}
		\end{align}
		In Lemma \ref{lemma:RenyiPolar1}, by letting $X_1=X_2=U^i$, $Y_1=Y_2=(Y^{1:2^n},U^{1:i-1})$, $U_1=U^{2i-1}$, and $U_2=U^{2i}$, we can see that (\ref{Converge-1}) and (\ref{Converge-2}) force (\ref{Proof-7}) to hold with equality for $i\in[2^n]$ a.e. as $n\rightarrow \infty$. From Proposition \ref{Prop:Minkowski} we know that (\ref{Proof-7}) holds with equality only in the following two cases:
		\begin{itemize}
			\item (Case 1) $P_{Y_1,X_1}(y_1,0)=0$ or $P_{Y_1,X_1}(y_1,1)=0$ for all effective $y_1\in\mathcal{Y}$.
			\item (Case 2) $P_{Y_2,X_2}(y_2,0)=P_{Y_2,X_2}(y_2,1)$ for all effective $y_2\in\mathcal{Y}$.
		\end{itemize}
	    The effective elements of $\mathcal{Y}$ (denoted by $\mathcal{Y}_{e}\subset \mathcal{Y}$) mean that
	    \begin{align*}
	    \frac{\sum_{\{x,y\}\in\mathcal{X}\times\mathcal{Y}_{e}}{P_{X,Y}(x,y)^{\alpha}} }{\sum_{y\in\mathcal{Y}_{e}}P_Y(y)^{\alpha}}\rightarrow \frac{\sum_{\{x,y\}\in\mathcal{X}\times\mathcal{Y}}{P_{X,Y}(x,y)^{\alpha}} }{\sum_{y\in\mathcal{Y}}P_Y(y)^{\alpha}}
	    \end{align*}
	    as $n\rightarrow \infty$. A more detailed discussion about it will be given in the next subsection. Since $P_{Y_1,X_1}$ and $P_{Y_2,X_2}$ are identical in the recursive transformation, the joint distributions of $P_{\mathbf{Y}^i,U^i}$ ($i\in [2^n]$) tend to either Case 1 or Case 2 a.e. as $n\rightarrow \infty$, where $\mathbf{Y}^i=(Y^{1:2^n},U^{1:i-1})$.
		It is easy to verify that $H_{\alpha}(U^i|\mathbf{Y}^i)$ equals 0 in Case 1 and 1 in Case 2. This shows that $H_{n}$ converges a.e. to 0 and 1 as $n\rightarrow \infty$. 
		
		The convergence result together with the chain rule equality of (\ref{RenyiPolar-a3}) imply that the fraction of $\{i:H(n,i)\in (1-\delta,1]\}$ goes to $H_{\alpha}(X|Y)$ as $n\rightarrow \infty$.
		
	\end{proof}

    Now we can answer the question raised at the beginning of this paper. From the definition of the max-entropy $H_{0}(X)$ and (\ref{RenyiCondChain}) we know that
    \begin{align}
    H_{0}(X|Y)
    =\log |\mathcal{X}\times\mathcal{Y}|-\log |\mathcal{Y}|=1,
    \end{align}
    provided that the probabilities $P_{\mathbf{Y}^i,U^i}(\mathbf{y}^i,u^i)$ are nonzero. Since $P_{\mathbf{Y}^i,U^i}(\mathbf{y}^i,u^i)$ never really become 0 by the polar transformation, we know that the fraction of sub-channels with $H_{0}(U^i|\mathbf{Y}^i)\rightarrow 0$ is always 0, which means that the synthetic sub-channels never really polarize to truly deterministic state. On the contrary, the min-conditional-entropy,
    \begin{align}
    H_{\infty}(X|Y)&=H_{\infty}(XY)-H_{\infty}(Y)\nonumber\\
    &=\log\frac{\max_{y\in\mathcal{Y}} p_y}{\max_{\{x,y\}\in\mathcal{X}\times\mathcal{Y}} p_{x,y}},
    \end{align}
    is not a constant in general. As will be explained in Section \ref{Sec:DiscussRenyi}, only truly uniform distribution yields $H_{\infty}(X|Y)=1$. Therefore, as $n\rightarrow \infty$, a fraction $H_{\infty}(X|Y)$ of the sub-channels will become truly completely noisy. This result can be concluded by the following corollary.
    
    \newtheorem{corollary}{Corollary}
    \begin{corollary}
    	As $n\rightarrow\infty$, the fraction of truly completely noisy sub-channels tends to $H_{\infty}(X|Y)$, while the fraction of truly noiseless sub-channels is 0.
    	
    \end{corollary}

    \subsection{An Example of Effective Elements}
    The reason we introduce the term of effective elements is that, without the "effective" in the definitions of Case 1 and Case 2,  the conditional R\'{e}nyi entropies in Case 1 and Case 2 of different orders should all be 0 and 1, respectively. In this subsection, we further discuss this issue. First, let us present an example to show that for a given $\alpha=\alpha_0$ and $\alpha'=\alpha_0+1$, $H_{\alpha}(X|Y)$ can be arbitrarily close to 1 while $H_{\alpha'}(X|Y)$ can be arbitrarily close to 0.
    
    Let $|\mathcal{Y}|=2^N\triangleq M$. Consider a joint distribution $P_{X,Y}$ such that a fraction $\frac{1}{L}$ of probability pairs $\big{(} P_{X,Y}(0,y),P_{X,Y}(1,y)\big{)}$ are completely deterministic with accumulated probability of $\frac{1}{N}$. Without loss of generality, we assume that $P_{X,Y}(0,y_i)=\frac{L}{NM}$ and $P_{X,Y}(1,y_i)=0$ for $i\in \mathcal{A}$, where $\mathcal{A}\subset [2^N]$ is the set of deterministic pairs. The rest $\frac{L-1}{L}$ fraction of probability pairs are completely uniform, i.e., $P_{X,Y}(0,y_i)=P_{X,Y}(1,y_i)=\frac{(N-1)L}{2NM(L-1)}$ for $i\in \mathcal{A}^C$. Then
    \begin{align}
    &~~~~H_{\alpha}(X|Y)\nonumber\\
    &=\frac{1}{1-\alpha}\log \frac{\sum_{i\in\mathcal{A}}{(\frac{L}{NM})^{\alpha}}+2^{1-\alpha}\sum_{i\in\mathcal{A}^C}{(\frac{(N-1)L}{NM(L-1)})^{\alpha}} }{\sum_{i\in\mathcal{A}}{(\frac{L}{NM})^{\alpha}}+\sum_{i\in\mathcal{A}^C}{(\frac{(N-1)L}{NM(L-1)})^{\alpha}}}\nonumber\\
    &=\frac{1}{1-\alpha}\log \frac{\frac{M}{L}{(\frac{L}{NM})^{\alpha}}+2^{1-\alpha}\frac{M(L-1)}{L}{(\frac{(N-1)L}{NM(L-1)})^{\alpha}} }{\frac{M}{L}{(\frac{L}{NM})^{\alpha}}+\frac{M(L-1)}{L}{(\frac{(N-1)L}{NM(L-1)})^{\alpha}}}\nonumber\\
    &=\frac{1}{1-\alpha}\log \frac{(L-1)^{\alpha-1}+2^{1-\alpha}(N-1)^{\alpha}}{(L-1)^{\alpha-1}+(N-1)^{\alpha}}.  \label{example-1}
    \end{align}
    For a considered $\alpha=\alpha_0>1$, let
    \begin{equation}
    L-1=2^{-1}(N-1)^{\frac{\alpha_0-0.5/\alpha_0}{\alpha_0-1}}.\label{example-L}
    \end{equation}
    Then we have
    \begin{align*}
    H_{\alpha_0}(X|Y)&=\frac{1}{1-\alpha_0}\log \frac{2^{1-\alpha_0}\big{[}(N-1)^{\alpha_0-\frac{0.5}{\alpha_0}}+(N-1)^{\alpha_0}\big{]}}{2^{1-\alpha_0}(N-1)^{\alpha_0-\frac{0.5}{\alpha_0}}+(N-1)^{\alpha_0}}.
    \end{align*}
    It is clear that as $N\rightarrow\infty$,$\frac{(N-1)^{\alpha_0}}{(N-1)^{\alpha_0-\frac{0.5}{\alpha_0}}}=(N-1)^{0.5/\alpha_0}\rightarrow \infty$, thus $H_{\alpha_0}(X|Y)\rightarrow 1$.
    
    Now let $\alpha'=\alpha_0+1$. From (\ref{example-1}) we have
    \begin{align*}
    &~~~~H_{\alpha'}(X|Y)\\
    &=\frac{1}{-\alpha_0}\log \frac{2^{-\alpha_0}(N-1)^{\frac{\alpha_0^2-0.5}{\alpha_0-1}}+2^{-\alpha_0}(N-1)^{\alpha_0+1}}{2^{-\alpha_0}(N-1)^{\frac{\alpha_0^2-0.5}{\alpha_0-1}}+(N-1)^{\alpha_0+1}}.
    \end{align*}
    In this case, as $N\rightarrow\infty$,$\frac{(N-1)^{\alpha_0+1}}{(N-1)^{\frac{\alpha_0^2-0.5}{\alpha_0-1}}}=(N-1)^{\frac{-0.5}{\alpha_0^2-1}}\rightarrow 0$, thus $H_{\alpha'}(X|Y)\rightarrow 0$.
    
    From (\ref{example-L}) we can see that $\frac{1}{L}\rightarrow 0$ as $N\rightarrow \infty$ if $\alpha_0>1$. Also, $\frac{1}{N}\rightarrow 0$ as $N\rightarrow \infty$. Therefore, we have shown a case when $H_{\alpha}(X|Y)$ and $H_{\alpha+1}(X|Y)$ can be completely opposite asymptotically. Fig. \ref{fig:RenyiExtreme} shows the example of $\alpha_0=2$. 
    \begin{figure}[tb]
    	\centering
    	\includegraphics[width=7cm]{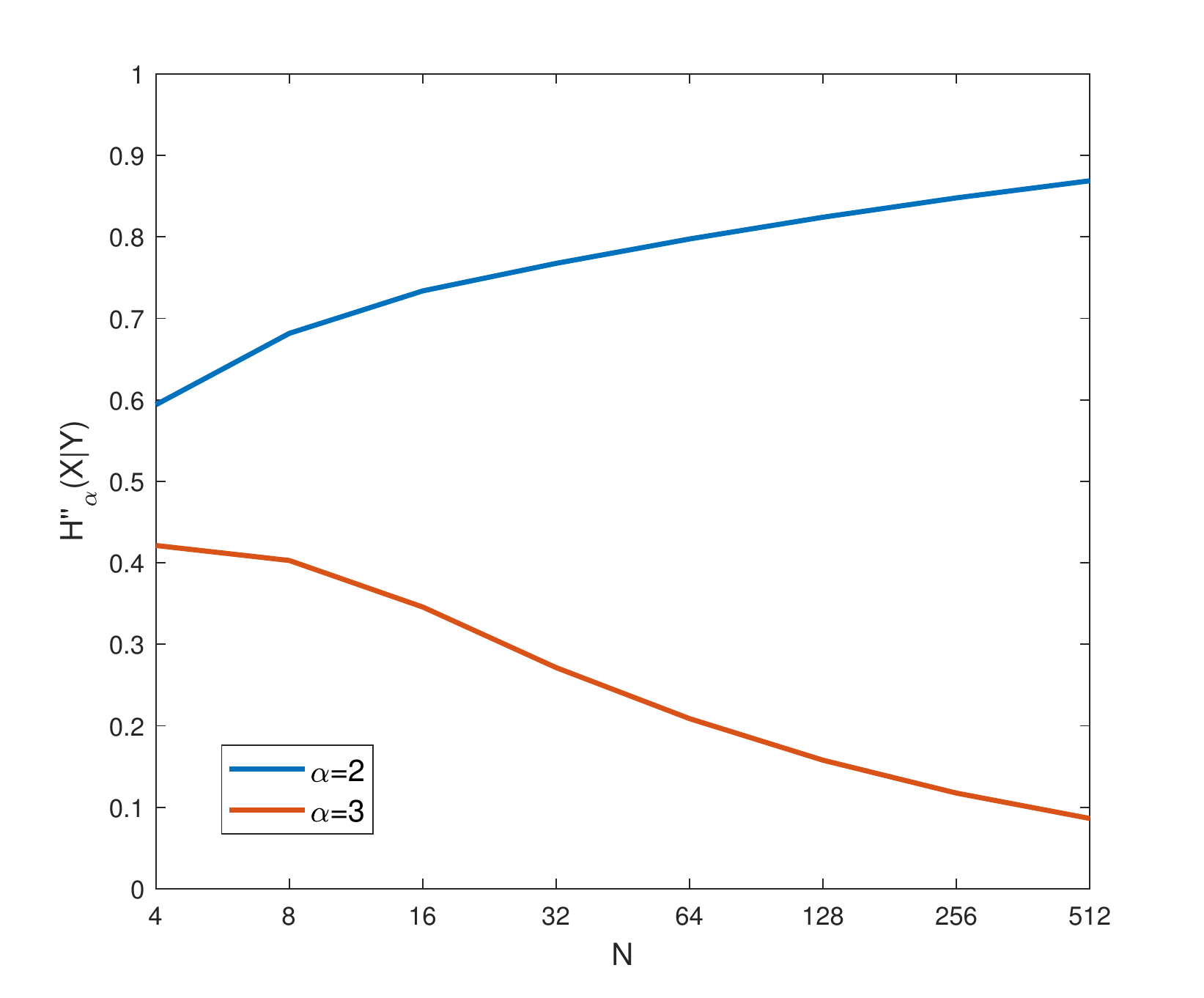}
    	\caption{conditional R\'{e}nyi entropies of different $\alpha$.} \label{fig:RenyiExtreme}
    \end{figure}

    We can similarly design such an example for $0<\alpha<1$. This shows that the extreme cases defined in the proof of Theorem \ref{Theorem:polarRenyi} are relative. Even if almost all probability pairs of a sub-channel are uniform (so that it may seem to be of Case 2), when powered by some $\alpha$, these probability pairs may have little impact on the value of $H_{\alpha}(X|Y)$ (so that the sub-channel is actually of Case 1), just as our example has shown. Thus, although we proved Theorem \ref{Theorem:polarRenyi} for any $\alpha\geq 0$ in a unified form, the criterion for judging whether a sub-channel converges to Case 1 or Case 2 depends on $\alpha$.

    \subsection{Numerical Results}
    
    \begin{figure}[tb]
    	\centering
    	\includegraphics[width=8cm]{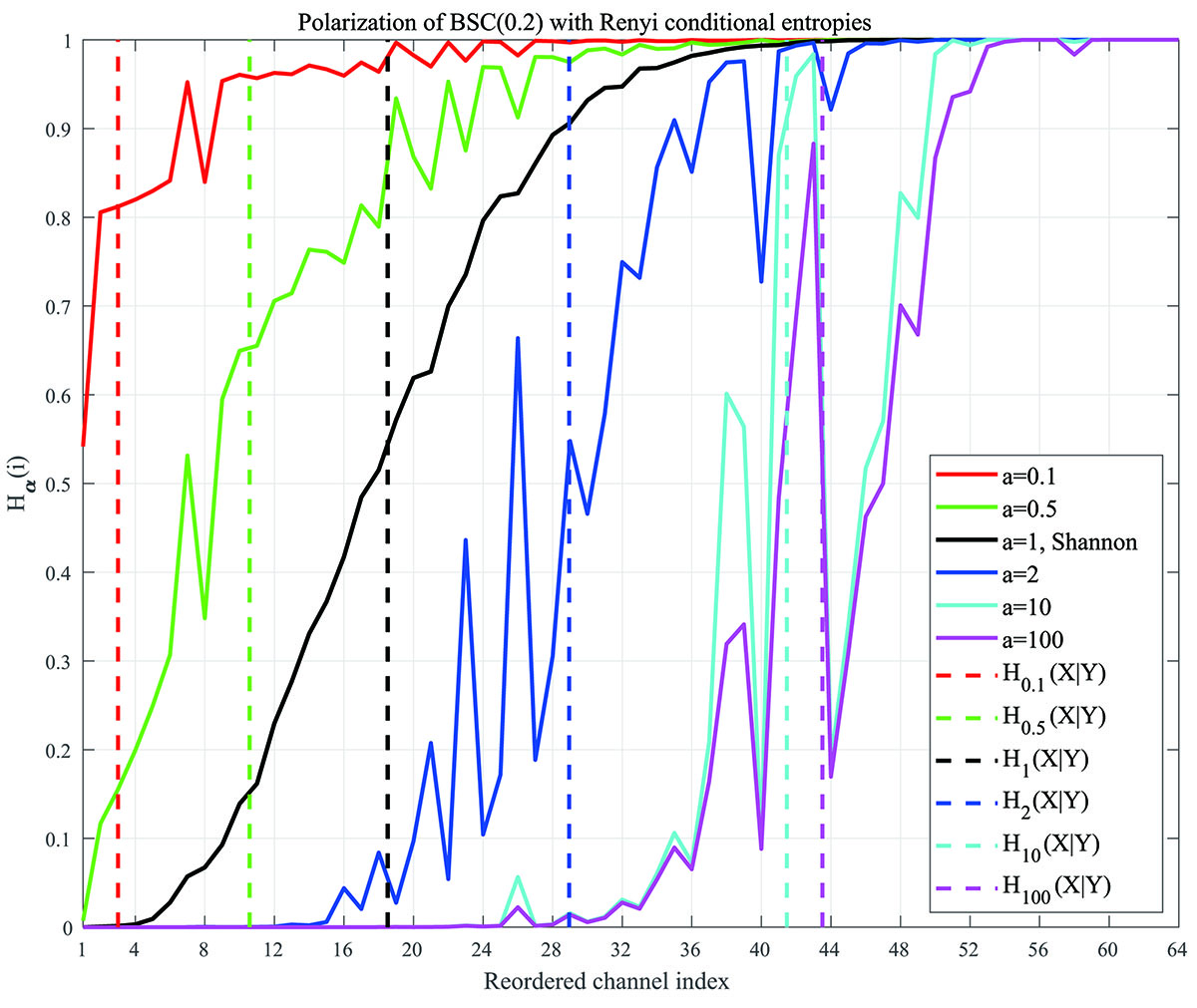}
    	\caption{Polarization of conditional R\'{e}nyi entropies.} \label{fig:PolarRenyi}
    \end{figure}
    In this subsection we present the polarization of a binary symmetric channel (BSC) with crossover probability 0.2 numerically. We calculate the joint distribution of each synthetic sub-channel and then compute its conditional R\'{e}nyi entropies of order 0.1, 0.5, 1, 2, 10 and 100, respectively. The result is shown in Fig. \ref{fig:PolarRenyi}. The channel indices are reordered so that the conditional Shannon entropies increase monotonically. The dash lines demonstrate the proportions of extremal sub-channels as $n\rightarrow\infty$. Although the considered block-length is not long, this figure clearly shows that the polarized sets vary with $\alpha$. We can also see that a relatively higher conditional Shannon entropy does not necessarily imply a relatively higher conditional R\'{e}nyi entropy of a different order.

	\section{A Discussion on R\'{e}nyi entropy}
	\label{Sec:DiscussRenyi}
	In this section, we further discuss how small the deviation from a uniform or deterministic distribution should be to make the conditional R\'{e}nyi entropy achieve 1 or 0. 
	Let $\bar{Q}_{X,Y}$ be a uniform distribution with respect to $X$, i.e., $\bar{Q}_{X,Y}(0,y)=\bar{Q}_{X,Y}(1,y)=\frac{1}{2}\bar{Q}_Y(y)$ for any $y\in \mathcal{Y}$, and further assume that $\bar{Q}_Y(y)=P_Y(y)$. Let $\tilde{Q}_{X,Y}$ be a deterministic distribution with respect to $X$, i.e., $\tilde{Q}_{X,Y}(0,y)=0, \tilde{Q}_{X,Y}(1,y)=\tilde{Q}_Y(y)$ or $\tilde{Q}_{X,Y}(1,y)=0, \tilde{Q}_{X,Y}(0,y)=\tilde{Q}_Y(y)$ for any $y\in \mathcal{Y}$, and further assume that $\tilde{Q}_Y(y)=P_Y(y)$. Then we have
	

	\begin{align*}
	H_{\alpha}(X|Y)
	&=\frac{1}{1-\alpha}\log\Bigg{\{} \frac{\sum_{\{x,y\}\in\mathcal{X}\times\mathcal{Y}}{P_{X,Y}(x,y)^{\alpha}} }{\sum_{\{x,y\}\in\mathcal{X}\times\mathcal{Y}}{\bar{Q}_{X,Y}(x,y)^{\alpha}}}\nonumber\\
	&~~~~~~~~\times \frac{2^{1-\alpha}\sum_{y\in\mathcal{Y}}\bar{Q}_Y(y)^{\alpha}}{\sum_{y\in\mathcal{Y}}P_Y(y)^{\alpha}}\Bigg{\}}\\
	&=1+\frac{1}{1-\alpha}\log \frac{\sum_{\{x,y\}\in\mathcal{X}\times\mathcal{Y}}{P_{X,Y}(x,y)^{\alpha}} }{\sum_{\{x,y\}\in\mathcal{X}\times\mathcal{Y}}{\bar{Q}_{X,Y}(x,y)^{\alpha}}},
	\end{align*}
	and	similarly
	\begin{align*}
	H_{\alpha}(X|Y)
	&=\frac{1}{1-\alpha}\log \frac{\sum_{\{x,y\}\in\mathcal{X}\times\mathcal{Y}}{P_{X,Y}(x,y)^{\alpha}} }{\sum_{\{x,y\}\in\mathcal{X}\times\mathcal{Y}}{\tilde{Q}_{X,Y}(x,y)^{\alpha}}}.
	\end{align*}
	
	For the uniform distribution case, assume that $P_{X,Y}(0,y_i)=\frac{1}{2}\bar{Q}_Y(y_i)+\delta_i, P_{X,Y}(1,y_i)=\frac{1}{2}\bar{Q}_Y(y_i)-\delta_i$, where $\delta_i\ll\bar{Q}_Y(y_i)$. Denote $|\mathcal{Y}|=M$. Define
	\begin{align}
	\Delta&=\frac{\sum_{\{x,y\}\in\mathcal{X}\times\mathcal{Y}}{P_{X,Y}(x,y)^{\alpha}} }{\sum_{\{x,y\}\in\mathcal{X}\times\mathcal{Y}}{\bar{Q}_{X,Y}(x,y)^{\alpha}}}-1\\
	&\approx\frac{2\alpha(\alpha-1)\sum_{i=1}^{M}\delta_i^{2}\bar{Q}^{\alpha-2}_Y(y_i) }{\sum_{i=1}^M{\bar{Q}_{Y}(y_i)^{\alpha}}}.
	\end{align}
	As a simple example, further assume that $\bar{Q}_{Y}$ is also uniform. Then to ensure that $H_{\alpha}(X|Y)\rightarrow 1$,
	\begin{align*}
	\alpha(\alpha-1)M^2\delta_i^2\rightarrow 0
	\end{align*} is required. As $\alpha\rightarrow \infty$, only truly uniform distribution yields $H_{\alpha}(X|Y)=1$.

	For the deterministic distribution case, without loss of generality, assume $\tilde{Q}_{X,Y}(0,y_i)=0, \tilde{Q}_{X,Y}(1,y_i)=\tilde{Q}_Y(y_i)$, and $P_{X,Y}(0,y_i)=\delta_i, P_{X,Y}(1,y_i)=\tilde{Q}_Y(y_i)-\delta_i$ for any $y_i\in \mathcal{Y}$, where $\delta_i\ll\tilde{Q}_Y(y_i)$. Denote $|\mathcal{Y}|=M$. 
	Define
	\begin{align}
	\Delta&=\frac{\sum_{\{x,y\}\in\mathcal{X}\times\mathcal{Y}}{P_{X,Y}(x,y)^{\alpha}} }{\sum_{\{x,y\}\in\mathcal{X}\times\mathcal{Y}}{\tilde{Q}_{X,Y}(x,y)^{\alpha}}}-1\\
	&\approx\frac{\sum_{i=1}^{M}\Big{[}\delta_i^{\alpha}-\alpha\delta_i\tilde{Q}^{\alpha-1}_Y(y_i) \Big{]}}{\sum_{i=1}^M{\tilde{Q}_{Y}(y_i)^{\alpha}}}.
	\end{align}
	For $0<\alpha<1$, the difference between $\delta_i$ and $\tilde{Q}_{Y}(y_i)$ shrinks by the power-$\alpha$ operation. As $\alpha$ approaches 0, $\sum_{i=1}^{M}\delta_i^{\alpha}$ gains more influence on $\Delta$. When $\alpha$ is small enough, we will have
	\begin{align}
	\Delta\approx \frac{\sum_{i=1}^{M}\delta_i^{\alpha}}{\sum_{i=1}^M{\tilde{Q}_{Y}(y_i)^{\alpha}}}.
	\end{align} 
	In the extreme case when $\alpha=0$ and $\delta_i>0$ for all $i\in[M]$, we get $\Delta=1$, and $H_{\alpha}(X|Y)$ always equals 1. $\Delta$ equals 0 if and only if $\delta_i=0$ for all $i\in[M]$.

	\section{Concluding Remarks}
	This work has revealed the polarization phenomenon of conditional R\'{e}nyi entropies. To apply the results to specific problems, much work has yet to be done. For example, estimating conditional R\'{e}nyi entropies accurately can be a hard problem when $N$ grows large, and existing approximation methods for polar code constructions may not be directly applied. For another example, the polarization rate of conditional R\'{e}nyi entropies has not been touched in this paper. We will leave them for future research.

	\bibliographystyle{IEEEtran}
	\bibliography{Polar_Renyi}

\begin{thebibliography}{1}
\providecommand{\url}[1]{#1}
\csname url@samestyle\endcsname
\providecommand{\newblock}{\relax}
\providecommand{\bibinfo}[2]{#2}
\providecommand{\BIBentrySTDinterwordspacing}{\spaceskip=0pt\relax}
\providecommand{\BIBentryALTinterwordstretchfactor}{4}
\providecommand{\BIBentryALTinterwordspacing}{\spaceskip=\fontdimen2\font plus
\BIBentryALTinterwordstretchfactor\fontdimen3\font minus
  \fontdimen4\font\relax}
\providecommand{\BIBforeignlanguage}[2]{{%
\expandafter\ifx\csname l@#1\endcsname\relax
\typeout{** WARNING: IEEEtran.bst: No hyphenation pattern has been}%
\typeout{** loaded for the language `#1'. Using the pattern for}%
\typeout{** the default language instead.}%
\else
\language=\csname l@#1\endcsname
\fi
#2}}
\providecommand{\BIBdecl}{\relax}
\BIBdecl

\bibitem{arikan2009channel}
E.~Ar{\i}kan, ``Channel polarization: A method for constructing
  capacity-achieving codes for symmetric binary-input memoryless channels,''
  \emph{IEEE Transactions on Information Theory}, vol.~55, no.~7, pp.
  3051--3073, 2009.

\bibitem{arikan2010source}
------, ``Source polarization,'' in \emph{2010 IEEE International Symposium on
  Information Theory}, 2010, pp. 899--903.

\bibitem{arikan2016varentropy}
------, ``Varentropy decreases under the polar transform,'' \emph{IEEE
  Transactions on Information Theory}, vol.~62, no.~6, pp. 3390--3400, 2016.

\bibitem{renyi1961measures}
A.~R{\'e}nyi, ``On measures of entropy and information,'' Hungarian Academy of
  Sciences Budapest Hungary, Tech. Rep., 1961.

\bibitem{bloch2011physical}
M.~Bloch and J.~Barros, \emph{Physical-layer security}.\hskip 1em plus 0.5em
  minus 0.4em\relax Cambridge University Press, 2011.

\bibitem{jizba2004world}
P.~Jizba and T.~Arimitsu, ``The world according to {R}{\'e}nyi: thermodynamics
  of multifractal systems,'' \emph{Annals of Physics}, vol. 312, no.~1, pp.
  17--59, 2004.

\bibitem{golshani2009some}
L.~Golshani, E.~Pasha, and G.~Yari, ``Some properties of {R}{\'e}nyi entropy
  and {R}{\'e}nyi entropy rate,'' \emph{Information Sciences}, vol. 179,
  no.~14, pp. 2426--2433, 2009.

\bibitem{Alsan2014E0}
M.~Alsan and E.~Telatar, ``Polarization improves ${E}_{0}$,'' \emph{IEEE
  Transactions on Information Theory}, vol.~60, no.~5, pp. 2714--2719, 2014.

\bibitem{arikan2014martingale}
E.~Ar{\i}kan, ``A note on polarization martingales,'' in \emph{2014 IEEE
  International Symposium on Information Theory}, 2014, pp. 1466--1468.

\end{thebibliography}
	
\end{document}